\documentclass[conference]{IEEEtran}

\usepackage[mathscr]{eucal}
\usepackage[cmex10]{amsmath}
\usepackage{epsfig,epsf,psfrag}
\usepackage{amssymb,amsmath,amsthm,amsfonts,latexsym}
\usepackage{amsmath,graphicx,bm,xcolor,url}
\usepackage[caption=false]{subfig} 
\usepackage{fixltx2e}
\usepackage{array}
\usepackage{verbatim}
\usepackage{bm}
\usepackage{algorithmic, cite}
\usepackage{algorithm}
\usepackage{verbatim}
\usepackage{textcomp}
\usepackage{mathrsfs}
\usepackage{epstopdf}

\catcode`~=11 \def\UrlSpecials{\do\~{\kern -.15em\lower .7ex\hbox{~}\kern .04em}} \catcode`~=13 

\allowdisplaybreaks[1]
 
\newcommand{\nn}{\nonumber}

\newcommand{\calC}{\mathcal{C}}

\newcommand{\calL}{\mathcal{L}}

\newcommand{\calN}{\mathcal{N}}

\newcommand{\calP}{\mathcal{P}}

\newcommand{\calU}{\mathcal{U}}
\newcommand{\calV}{\mathcal{V}}

\newcommand{\calX}{\mathcal{X}}
\newcommand{\calY}{\mathcal{Y}}

\newcommand{\ba}{\mathbf{a}}
\newcommand{\bA}{\mathbf{A}}
\newcommand{\bb}{\mathbf{b}}

\newcommand{\bD}{\mathbf{D}}

\newcommand{\bI}{\mathbf{I}}
\newcommand{\bj}{\mathbf{j}}

\newcommand{\bL}{\mathbf{L}}

\newcommand{\bp}{\mathbf{p}}

\newcommand{\bR}{\mathbf{R}}

\newcommand{\bT}{\mathbf{T}}
\newcommand{\bu}{\mathbf{u}}
\newcommand{\bU}{\mathbf{U}}
\newcommand{\bv}{\mathbf{v}}
\newcommand{\bV}{\mathbf{V}}

\newcommand{\bx}{\mathbf{x}}
\newcommand{\bX}{\mathbf{X}}
\newcommand{\by}{\mathbf{y}}
\newcommand{\bY}{\mathbf{Y}}
\newcommand{\bz}{\mathbf{z}}

\newcommand{\rmd}{\mathrm{d}}

\newcommand{\bbE}{\mathbb{E}}

\newcommand{\bbR}{\mathbb{R}}

\DeclareMathAlphabet{\mathbsf}{OT1}{cmss}{bx}{n}
\DeclareMathAlphabet{\mathssf}{OT1}{cmss}{m}{sl}

\DeclareSymbolFont{bsfletters}{OT1}{cmss}{bx}{n}  
\DeclareSymbolFont{ssfletters}{OT1}{cmss}{m}{n}
\DeclareMathSymbol{\bsfGamma}{0}{bsfletters}{'000}
\DeclareMathSymbol{\ssfGamma}{0}{ssfletters}{'000}
\DeclareMathSymbol{\bsfDelta}{0}{bsfletters}{'001}
\DeclareMathSymbol{\ssfDelta}{0}{ssfletters}{'001}
\DeclareMathSymbol{\bsfTheta}{0}{bsfletters}{'002}
\DeclareMathSymbol{\ssfTheta}{0}{ssfletters}{'002}
\DeclareMathSymbol{\bsfLambda}{0}{bsfletters}{'003}
\DeclareMathSymbol{\ssfLambda}{0}{ssfletters}{'003}
\DeclareMathSymbol{\bsfXi}{0}{bsfletters}{'004}
\DeclareMathSymbol{\ssfXi}{0}{ssfletters}{'004}
\DeclareMathSymbol{\bsfPi}{0}{bsfletters}{'005}
\DeclareMathSymbol{\ssfPi}{0}{ssfletters}{'005}
\DeclareMathSymbol{\bsfSigma}{0}{bsfletters}{'006}
\DeclareMathSymbol{\ssfSigma}{0}{ssfletters}{'006}
\DeclareMathSymbol{\bsfUpsilon}{0}{bsfletters}{'007}
\DeclareMathSymbol{\ssfUpsilon}{0}{ssfletters}{'007}
\DeclareMathSymbol{\bsfPhi}{0}{bsfletters}{'010}
\DeclareMathSymbol{\ssfPhi}{0}{ssfletters}{'010}
\DeclareMathSymbol{\bsfPsi}{0}{bsfletters}{'011}
\DeclareMathSymbol{\ssfPsi}{0}{ssfletters}{'011}
\DeclareMathSymbol{\bsfOmega}{0}{bsfletters}{'012}
\DeclareMathSymbol{\ssfOmega}{0}{ssfletters}{'012}

\newcommand{\bLambda}{\bm{\Lambda}}

\newcommand{\bDelta}{\bm{\Delta}}

\newcommand{\eps}{\varepsilon}

\DeclareMathOperator{\var}{\mathsf{Var}}

\DeclareMathOperator{\cov}{\mathsf{Cov}}

\DeclareMathOperator{\rank}{rank}

\newcommand{\bzero}{\mathbf{0}}
\newcommand{\bone}{\mathbf{1}}

\newtheorem{theorem}{Theorem}

\newtheorem{definition}{Definition}

\newcommand{\qednew}{\nobreak \ifvmode \relax \else
      \ifdim\lastskip<1.5em \hskip-\lastskip
      \hskip1.5em plus0em minus0.5em \fi \nobreak
      \vrule height0.75em width0.5em depth0.25em\fi}

\usepackage{flushend}
\usepackage{dsfont}

\newcommand{\Tlhat}{\hat{\bT}_{-}}
\newcommand{\Tlhati}{\hat{T}_{-,1}}
\newcommand{\Tlhatii}{\hat{T}_{-,2}}
\newcommand{\Trhat}{\hat{\bT}_{+}}
\newcommand{\Trhati}{\hat{T}_{+,1}}
\newcommand{\Trhatii}{\hat{T}_{+,2}}
\newcommand{\Tleft}{{\bT}_{-}}
\newcommand{\Tright}{{\bT}_{+}}
\newcommand{\cl}{\mathrm{cl}}
\newcommand{\openone}{\mathds{1}}

\begin{document} 

\title{Second-Order Asymptotics for the Discrete  \\ Memoryless MAC with Degraded Message Sets} 

\author{
 \IEEEauthorblockN{Jonathan Scarlett}
  \IEEEauthorblockA{Laboratory for Information and Inference Systems\\
    \'Ecole Polytechnique F\'ed\'erale de Lausanne \\
    Email: \url{jmscarlett@gmail.com} }
  \and 
  \IEEEauthorblockN{Vincent Y. F. Tan}
  \IEEEauthorblockA{Dept.\ of Elec.\ and Comp.\ Eng.\ and Dept.\ of Mathematics,  \\
       National University of Singapore \\
   Email: \url{vtan@nus.edu.sg}}
} 
 

\maketitle
 
\begin{abstract}
    This paper studies  
    the second-order asymptotics of the discrete memoryless multiple-access channel with 
    degraded message sets. For a fixed  average error probability $\eps \in (0,1)$ 
    and an arbitrary point on the boundary of the capacity region, we characterize 
    the speed of convergence of rate pairs that converge to that point 
    for codes that have asymptotic error probability no larger than $\eps$,
    thus complementing an analogous result given previously for the Gaussian setting.
\end{abstract} 
 
\section{Introduction}


In recent years, there has been great interest in characterizing the fixed-error  
asymptotics (e.g. dispersion, the Gaussian approximation) 
of source coding and channel coding problems, and the behavior is well-understood
for a variety of single-user settings \cite{Strassen,Hayashi09,PPV10}.  
On the other hand, analogous studies of multi-user
problems have generally had significantly less success, with the main exceptions being
Slepian-Wolf coding \cite{TK12,Nom14}, the Gaussian interference channel with strictly very strong
interference \cite{LeTan14}, and the Gaussian multiple-access channel (MAC) with degraded 
message sets~\cite{ScaTan13}.

In this paper, we complement our work on the latter problem by considering its  discrete
counterpart. By obtaining matching achievability and converse results, we provide
the first complete characterization of the second-order 
asymptotics for a discrete channel-type network information theory problem. 

\subsection{System Setup}

We consider the two-user discrete memoryless MAC (DM-MAC) with degraded message 
sets \cite[Ex.~5.18]{elgamal}, with input alphabets $\calX_1$ and $\calX_2$ and 
output alphabet $\calY$.  As usual, there are two messages $m_1$ and $m_2$,
equiprobable on the sets $\{1,\dotsc,M_1\}$ and $\{1,\dotsc,M_2\}$ respectively. 
The first user knows both messages, whereas the second user only knows $m_2$.  Given
these messages, the users transmit the codewords $\bx_1(m_1,m_2)$ and $\bx_2(m_2)$ from their 
respective codebooks, and the decoder receives a noisy output sequence which is generated according
to the memoryless transition law $W^n(\by|\bx_1,\bx_2) = \prod_{i=1}^{n}W(y_i|x_{1,i},x_{2,i})$.  
An estimate $(\hat{m}_1,\hat{m}_2)$ is formed, and an error is said to have occurred if
$(\hat{m}_1,\hat{m}_2) \ne (m_1,m_2)$.

The capacity region $\calC$ is given by the set of rate pairs $(R_1,R_2)$ satisfying \cite[Ex. 5.18]{elgamal}
\begin{align}
    R_1       \le I(X_1;Y|X_2) \label{eqn:cap_region1} \\
    R_1 + R_2 \le I(X_1,X_2;Y) \label{eqn:cap_region2}
\end{align} 
for some  input joint distribution $P_{X_1X_2}$, where the mutual information quantities are 
with respect to $P_{X_1X_2}(x_1,x_2)W(y|x_1,x_2)$.  The achievability part is proved
using superposition coding.

We formulate the second-order asymptotics according to the following definition \cite{Nom14}.

\begin{definition}[Second-Order Coding Rates] \label{def;second}
    Fix $\eps\in (0,1)$, and let $(R_1^*, R_2^*)$ be a pair of rates on the boundary of $\calC$. A pair $(L_1, L_2)$ is {\em $(\eps,R_1^*,R_2^*)$-achievable} 
    if there exists a sequence of codes with length $n$, number of codewords
    for message $j=1,2$ equal to $M_{j,n}$, and average error probability $\eps_n$, such that  
    \begin{gather}
    \liminf_{n\to\infty}\frac{1}{\sqrt{n}}(\log M_{j,n} - nR_j^*) \ge L_j,\quad j = 1,2, \label{eqn:second_def0}\\
    \limsup_{n\to\infty}\eps_n \le \eps. \label{eqn:second_def}
    \end{gather}
    The {\em $(\eps,R_1^*,R_2^*)$-optimal second-order coding rate region} 
    $\calL(\eps;R_1^*,R_2^*) \subset\bbR^2$ is defined to be the closure of the 
    set of all $(\eps,R_1^*,R_2^*)$-achievable rate pairs  $(L_1, L_2)$. 
\end{definition}

Throughout the paper, we write non-asymptotic rates as 
$R_{1,n} := \frac{1}{n}\log{M_{1,n}}$ and $R_{2,n} := \frac{1}{n}\log{M_{2,n}}$.
Roughly speaking, the preceding definition is concerned with $\eps$-reliable codes 
such that $R_{j,n} \ge R_{j}^* + \frac{1}{\sqrt n}L_{j} + o\big(\frac{1}{\sqrt n}\big)$ for $j=1,2$.

We will also use the following standard definition: A rate pair $(R_1,R_2)$ is 
\emph{$(n,\eps)$-achievable} if there exists a length-$n$ code having an average 
error probability no higher than $\eps$, and whose rate is at least $R_{j}$
for message $j=1,2$.

\subsection{Notation}

Except where stated otherwise,\footnote{For example, the vectors in \eqref{eqn:rate_vec}--\eqref{eqn:L_vec},  do not adhere to this convention.} 
the $i$-th entry of a vector 
(e.g. $\by$) is denoted using a subscript (e.g. $y_i$).
For two vectors of the same length $\ba,\mathbf{b}\in\bbR^d$, 
the notation $\ba\le\bb$ means that $a_j\le b_j$ for 
all $j$. The notation $\calN(\bu;\bm{\mu}, \bLambda)$ 
denotes the  multivariate Gaussian probability density function 
(pdf) with mean $\bm{\mu}$ and covariance $\bLambda$. 
We use the standard asymptotic notations $O(\cdot)$, $o(\cdot)$,
$\Theta(\cdot)$, and $\omega(\cdot)$.
All logarithms have base $e$, and all rates have units of nats.
The closure operation is denoted by $\cl(\cdot)$.

The set of all probability distributions on an alphabet $\calX$ 
is denoted by $\calP(\calX)$, and the set of all types \cite[Ch. 2]{Csi97}
is denoted by $\calP_n(\calX)$.  For a given type $Q_X \in \calP_n(\calX)$,
we define the type class $T^n(Q_X)$ to be the set of sequences having type $Q_X$.
Similarly, given a conditional type $Q_{Y|X}$ and a sequence $\bx \in T^n(Q_X)$,  
we define $T_{\bx}^n(Q_{Y|X})$ to be the set of sequences $\by$ such that
$(\bx,\by) \in T^n(Q_X \times Q_{Y|X})$.

\section{Main Result}

\subsection{Preliminary Definitions}

Given the rate pairs $(R_{1,n}, R_{2,n})$ and $(R_1^*,R_2^*)$, we define
\begin{align}
    \bR_n := \begin{bmatrix}
    R_{1,n} \\ R_{1,n} + R_{2,n}
    \end{bmatrix}, \quad 
    \bR^* := \begin{bmatrix}
    R_1^* \\ R_1^* + R_2^*
    \end{bmatrix} \label{eqn:rate_vec}
\end{align}
Similarly, given the second-order rate pair $(L_1,L_2)$, we write
\begin{align}
    \bL := \begin{bmatrix}
    L_1 \\ L_1 + L_2 \label{eqn:L_vec}
    \end{bmatrix}
\end{align}
Given a joint input distribution $P_{X_1X_2}\in\calP(\calX_1\times\calX_2)$, we define 
$P_{X_1X_2Y} := P_{X_1X_2} \times W$, and denote the induced
marginals by $P_{Y|X_1}$, $P_Y$, etc.  We define the following {\em information density vector}, 
which implicitly depends on $P_{X_1X_2}$:
\begin{align} 
    \bj(x_1, x_2  , y) &:= \begin{bmatrix}
    j_1(x_1, x_2  , y)   &  j_{12}(x_1, x_2,y)
    \end{bmatrix}^T  \\
    &= \begin{bmatrix}
    \log\dfrac{W(y|x_1,x_2)}{P_{Y|X_2}(y|x_2)}  & \log\dfrac{W(y|x_1,x_2)}{P_{Y}(y)}  
    \end{bmatrix}^T.
    \label{eqn:info_dens}
\end{align} 
The mean and conditional covariance matrix are given by
\begin{align}
    \bI(P_{X_1X_2}) &=\bbE\big[\, \bj(X_1,X_2,Y)\big],  \label{eqn:mean_v}\\
    \bV(P_{X_1X_2}) &=\bbE\big[\cov \big(\bj(X_1,X_2,Y) \, \big|\, X_1,X_2 \big)\big]. \label{eqn:cov_v}
\end{align}
Observe that the entries of $\bI(P_{X_1X_2})$ are the mutual informations
appearing in \eqref{eqn:cap_region1}--\eqref{eqn:cap_region2}.  We write 
the entries of $\bI$ and $\bV$ using subscripts as follows:
\begin{align}
    \bI(P_{X_1X_2}) &= \begin{bmatrix}
    I_1(P_{X_1X_2}) \\ I_{12}(P_{X_1X_2})
    \end{bmatrix} \label{eqn:mi_vec}, \\
    \bV(P_{X_1X_2}) &= \begin{bmatrix}
    V_1(P_{X_1X_2})  & V_{1,12}(P_{X_1X_2}) \\ V_{1,12}(P_{X_1X_2})  & V_{12}(P_{X_1X_2}) \end{bmatrix}, \label{eqn:inf_disp_matr}
\end{align}
For a given point $(z_1, z_2) \in \bbR^2$ and a positive semi-definite matrix $\bV$, 
we define the multivariate Gaussian cumulative distribution function (CDF)
\begin{align}
    \Psi(z_1, z_2;\bV) :=\int_{-\infty}^{z_2}\int_{-\infty}^{z_1}\calN(\bu;\bzero,\bV)\,\rmd \bu,
\end{align}
and for a given $\eps\in (0,1)$, we define the corresponding ``inverse" set
\begin{align}
    \Psi^{-1}(\bV,\eps):=\left\{ (z_1, z_2)\in\bbR^2:\Psi(-z_1,- z_2;\bV) \ge  1-\eps\right\}. \label{eqn:psiinv}
\end{align}
Similarly, we let $\Phi(\cdot)$ denote the standard Gaussian CDF, and we denote
its functional inverse by $\Phi^{-1}(\cdot)$.  Moreover, we let 
\begin{align}
    \Pi(R_1^*,R_2^*) := \big\{ P_{X_1X_2} \,:\,   \bI(P_{X_1X_2}) \ge \bR^* \big\} \label{eqn:def_Pi}
\end{align}
be the set of input distributions achieving a given point $(R_1^*,R_2^*)$ of 
the boundary of $\calC$. Note that in contrast with the single-user setting 
\cite{Strassen, PPV10, Hayashi09}, this definition uses an inequality rather
than an equality, as one of the mutual information quantities may be 
strictly larger than the corresponding entry of $\bR^*$ and yet be 
first-order optimal. For example, assuming that the capacity region on the
left of Figure \ref{fig:tangents} is achieved by a single input distribution,
all points $(R_1^*,R_2^*)$ on the vertical boundary satisfy $I_{12}(P_{X_1 X_2}) > R_1^*+R_2^*$.

The preceding definitions are analogous to those appearing in previous works such
as \cite{TK12}, while the remaining definitions are somewhat less standard.  Given
the boundary point $(R_1^*,R_2^*)$, we let
$\Tlhat := \Tlhat(R_1^*,R_2^*)$ and $\Trhat := \Trhat(R_1^*,R_2^*)$ denote the
left and right unit tangent vectors along the boundary of $\calC$ in $(R_1,R_2)$
space; see Figure \ref{fig:tangents} for an illustration. Furthermore, we define
\begin{align}
    \Tleft := \begin{bmatrix} \Tlhati \\ \Tlhati + \Tlhatii \end{bmatrix}, \quad
    \Tright := \begin{bmatrix} \Trhati \\ \Trhati + \Trhatii \end{bmatrix}.
\end{align}
It is understood that $\Tlhat$ and $\Tleft$ (respectively, $\Trhat$ and $\Tright$)
are undefined when $R_1^* = 0$ (respectively, $R_2^* = 0$).  As is observed
in Figure \ref{fig:tangents}, we have $\Tlhat=-\Trhat$ on the curved
and straight-line parts of $\calC$, and $\Tlhat\ne-\Trhat$ when
there is a sudden change in slope (e.g. at a corner point).

\begin{figure}[t]
\centering
\includegraphics[width=0.95\columnwidth]{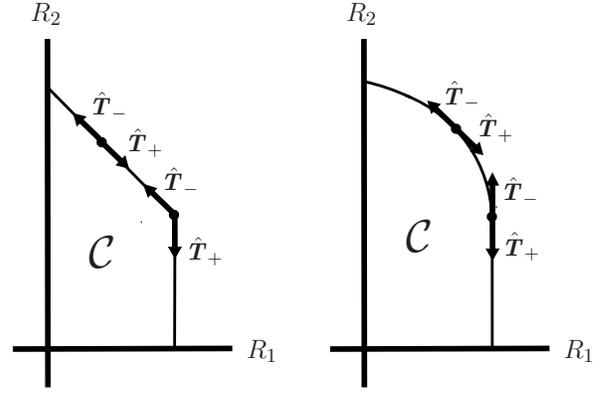}

\vspace{-3mm}
\caption{Unit tangent vectors $\Tlhat$ and $\Trhat$ for two boundary points $(R_1^*,R_2^*)$ of two hypothetical capacity regions.} 
\vspace{-3mm}

\label{fig:tangents}
\end{figure}

The following set of vectors can be thought of as those that point strictly inside  
$\calC$ when placed at $(R_1^*,R_2^*)$:
\begin{multline}
    \hat{\calV}(R_1^*,R_2^*) := \\ \big\{ \bv\in\bbR^2 \,:\, (R_1^*,R_2^*) + \alpha \bv \in \calC \text{ for some } \alpha>0 \big\}.
\end{multline}
Using this definition, we set
\begin{align}
    \calV(R_1^*,R_2^*) := \cl\bigg( \bigcup_{(v_1,v_2) \in \hat{\calV}(R_1^*,R_2^*)}\big\{ (v_1,v_1+v_2) \big\}\bigg). \label{eqn:setV}
\end{align}
Due to the closure operation, it is readily verified that $\Tleft \in \calV$ and
$\Tright \in \calV$.

\subsection{Statement of Main Result}

For a given boundary point $(R_1^*,R_2^*)$ and input distribution $P_{X_1X_2} \in \Pi(R_1^*,R_2^*)$, 
we define the set $\calL_0(\eps;R_1^*,R_2^*,P_{X_1X_2})$ separately for the following three cases:

\emph{(i)} If $R_1^* = I_1(P_{X_1X_2})$ and  $R_1^* + R_2^* < I_{12}(P_{X_1X_2})$, then
    \begin{align}
        \calL_0 = \Big\{ (L_1,L_2) \,:\, L_1 \le \sqrt{V_1(P_{X_1X_2})}\Phi^{-1}(\eps) \Big\} \label{eqn:L0_1}
    \end{align}
        
\emph{(ii)} If $R_1^* < I_1(P_{X_1X_2})$ and  $R_1^* + R_2^* = I_{12}(P_{X_1X_2})$, then
    \begin{align}
        \hspace*{-0.1ex}\calL_0 = \Big\{ (L_1,L_2) \,:\, L_1+L_2 \le \sqrt{V_{12}(P_{X_1X_2})}\Phi^{-1}(\eps) \Big\} \label{eqn:L0_12}
    \end{align}
        
\emph{(iii)} If $R_1^* = I_1(P_{X_1X_2})$ and  $R_1^* + R_2^* = I_{12}(P_{X_1X_2})$, then
    \begin{align}
        \calL_0 = \Big\{ (L_1,L_2) \,:\, \bL \in \bigcup_{\beta \ge 0} \big\{ \beta\Tleft + \Psi^{-1}(\bV(P_{X_1X_2}),\eps) \big\} \Big\} & \nn \\
            \cup~\Big\{ (L_1,L_2) \,:\, \bL \in \bigcup_{\beta \ge 0} \big\{ \beta\Tright + \Psi^{-1}(\bV(P_{X_1X_2}),\eps) \big\} \Big\}  &, \label{eqn:L0_both}
    \end{align}
where the first (respectively, second) set in the union is understood to be empty
when $R_1^*=0$ (respectively, $R_2^*=0$).
    
We are now in a position to state our main result.

\begin{theorem} \label{thm:local}
    For any point $(R_1^*,R_2^*)$ on the boundary of the capacity region, and any $\eps\in(0,1)$, we have
    \begin{equation}
        \calL(\eps;R_1^*,R_2^*) \!=\! \bigcup_{P_{X_1X_2} \in \Pi(R_1^*,R_2^*)} \calL_0(\eps;R_1^*,R_2^*,P_{X_1X_2}). \label{eqn:local}
    \end{equation}
\end{theorem}
\begin{proof}
    See Section \ref{sec:prf_local}.
\end{proof}

Suppose that $\calX_1=\emptyset$ and $(R_1^*, R_2^*)=(0,C)$, where $C := \max_{P_{X_2}}I(P_{X_2} , W)$ and $W:\calX_2\to\calY$.  Clearly $L_1$ plays no role, and Theorem \ref{thm:local} states that the achievable values of 
$L_2$ are precisely those in  the set 
\begin{equation}
 \calL_2(\eps  ):=    \bigcup_{P_{X_2} \in\Pi} \Big\{L_2: L_2 \le  \sqrt{V(P_{X_2}) }\Phi^{-1}(\eps) \Big\},
\end{equation}
where $\Pi := \{P_{X_2}  :  I(P_{X_2}, W) = C \}$, and $V(\cdot) := V_{12}(\cdot)$ is the conditional information variance \cite{PPV10}.
Letting $L^*  := \sup\calL_2(\eps  )$ be the {\em second-order coding 
rate} \cite{Hayashi09} of the discrete memoryless channel (DMC) $W:\calX_2\to\calY$, we readily obtain
\begin{equation}
    L^*  = \left\{  \begin{array}{cc}
    \sqrt{\min_{P_{X_2} \in\Pi} V(P_{X_2}) } \, \Phi^{-1}(\eps) &\eps < \frac12 \\
    \sqrt{ \max_{P_{X_2} \in\Pi} V(P_{X_2})  }\,\Phi^{-1}(\eps) &\eps \ge \frac12.
    \end{array}
    \right.
\end{equation}
Thus, our main result reduces to the classical result of Strassen~\cite[Thm.~3.1]{Strassen}  
for the single-user setting
(see also \cite{Hayashi09,PPV10}). This illustrates the necessity of the 
set $\Pi(R_1^*, R_2^*)$ in the characterization of $\calL(\eps;R_1^*,R_2^*)$ 
in Theorem~\ref{thm:local}. Such a set is not needed in the Gaussian setting~\cite{ScaTan13}, 
as every boundary point is achieved uniquely by a single multivariate Gaussian distribution. 
Another notable difference in Theorem \ref{thm:local} compared to \cite{ScaTan13} is the 
use of left and right tangent vectors instead of a single derivative vector.

Both of the preceding differences were also recently observed in an achievability 
result for the standard MAC \cite{Scarlett13b_arxiv}.  However, no converse results
were given in \cite{Scarlett13b_arxiv}, and the main novelty of the present paper is in the converse proof.

It is not difficult to show that $\calL$ equals a half-space whenever 
$\Tlhat=-\Trhat$, as was observed in \cite{ScaTan13,Scarlett13b_arxiv}.
A less obvious fact is that the unions over $\beta$ in \eqref{eqn:L0_both}
can be replaced by $\beta=0$ whenever the corresponding input distribution
$P_{X_1X_2}$ achieves all of the boundary points in a neighborhood
of $(R_1^*,R_2^*)$. We refer the reader to \cite{ScaTan13,Scarlett13b_arxiv} 
for further discussions and illustrative numerical examples.

\section{Proof of Theorem \ref{thm:local}} \label{sec:prf_local}

Due to space constraints, we do not attempt to make the proof self-contained.
We avoid repeating the parts in common with \cite{Scarlett13b_arxiv,ScaTan13}, and
we focus on the most novel aspects.

\subsection{Achievability}

The achievability part of Theorem \ref{thm:local} is proved using a similar
(yet simpler) argument to that of the standard MAC given in \cite{Scarlett13b_arxiv}, 
so we only provide a brief outline.

We use constant-composition superposition coding with coded 
time sharing \cite[Sec. 4.5.3]{elgamal}.  We set $\calU := \{1,2\}$, fix a joint
distribution $Q_{UX_1X_2}$ (to be specified shortly), and let $Q_{UX_1X_2,n}$ be 
the closest corresponding joint type. We write the marginal distributions
in the usual way (e.g. $Q_{X_1|U,n}$).
We let $\bu$ be a deterministic time-sharing sequence with $nQ_{U,n}(1 )$ ones and
$nQ_{U,n}(2 )$ twos.  We first generate the $M_{2,n}$ codewords of user 2 independently according to the
uniform distribution on  $T_{\bu}^n(Q_{X_1|U,n})$.  For
each $m_2$, we generate $M_{1,n}$ codewords for user 1 conditionally independently according  
to the uniform distribution on $T_{\bu\bx_2}^n(Q_{X_1|X_2U,n})$, where $\bx_2$ is  
the codeword for user 2 corresponding to $m_2$.  

We fix $\beta\ge0$ and choose $Q_{UX_1X_2}$ such that $Q_U(1) = 1 - \frac{\beta}{\sqrt n}$
and $Q_U(2) = \frac{\beta}{\sqrt n}$, let $Q_{X_1X_2|U=1}$ be an input distribution $P_{X_1X_2}$ achieving
the boundary point of interest, and let $Q_{X_1X_2|U=2}$ be an input distribution $P'_{X_1X_2}$ achieving 
a different boundary point.
We define the shorthands $\bI := \bI(P_{X_1X_2})$, $\bV := \bV(P_{X_1X_2})$ and $\bI' := \bI(P'_{X_1X_2})$.
Using the generalized Feinstein bound given in \cite{ScaTan13} along with the multivariate 
Berry-Esseen theorem, we can use the arguments of \cite{Scarlett13b_arxiv} to conclude that
all rate pairs $(R_{1,n},R_{2,n})$ satisfying
\begin{equation}
    \bR_n \in \bI + \frac{1}{\sqrt n}\Big(\beta(\bI'-\bI) + \Psi^{-1}(\bV,\eps) \Big) + g(n)\bone \label{eq:global_ach}
\end{equation}
are $(n,\eps)$-achievable for some $g(n)=O\big(n^{1/4}\big)$ depending on $\eps$, 
$\beta$, $P_{X_1X_2}$ and $P'_{X_1X_2}$.  Note that this argument may require 
a reduction to a lower dimension for singular dispersion matrices; an analogous
reduction will be given in the converse proof below.

The achievability part of Theorem \ref{thm:local} now follows as in \cite{Scarlett13b_arxiv}.
In the cases corresponding to \eqref{eqn:L0_1}--\eqref{eqn:L0_12}, we eliminate one of
the two element-wise inequalities from \eqref{eq:global_ach} to obtain the desired result.
For the remaining case corresponding to \eqref{eqn:L0_both}, we obtain the first
(respectively, second) term in the union by letting $P_{X_1X_2}'$ achieve a boundary  
point approaching $(R_1^*,R_2^*)$ from the left (respectively, right).

\subsection{Converse}

The converse proof builds on that for the Gaussian case \cite{ScaTan13}, but contains
more new ideas compared to the achievability part.  We thus provide a more detailed
treatment.  

\subsubsection{A Reduction from Average Error to Maximal Error}

Using an identical argument to the Gaussian case \cite{ScaTan13} (which itself
builds on~\cite[Cor.~16.2]{Csi97}), we can show
that $\calL(\eps;R_1^*,R_2^*)$ is identical when the average error probability
is replaced by the maximal error probability in Def.~\ref{def;second}.
We may thus proceed by considering the maximal error probability.  Note that
neither this step nor the following step are possible for the standard MAC;
the assumption of degraded message sets is crucial.

\subsubsection{A Reduction to Constant-Composition Codes}

Using the previous step and the fact that the number of joint types 
on $\calX_1 \times \calX_2$ is polynomial in $n$, we can again follow an
identical argument to the Gaussian case \cite{ScaTan13} to show that 
$\calL(\eps;R_1^*,R_2^*)$ is unchanged when the codebook
is restricted to contain codeword pairs $(\bx_1,\bx_2)$ sharing a common joint
type.  We thus limit our attention to such codebooks; we denote the corresponding
sequence of joint types by $\{P_{X_1X_2,n}\}_{n\ge1}$.

\subsubsection{Passage to a Convergent Subsequence} \label{step:passage}

Since  $\calP(\calX_1 \times \calX_2)$ is compact, the sequence
$\{P_{X_1X_2,n}\}_{n\ge1}$ must have a convergent subsequence, say indexed by a
sequence $\{n_k\}_{k\ge1}$ of block lengths.
We henceforth limit our attention to values of $n$ on this subsequence.  To avoid cumbersome
notation, we continue writing $n$ instead of $n_k$. 
However, it should be understood that asymptotic
notations such as $O(\cdot)$ and $(\cdot)_n \to (\cdot)$ 
are taken with respect to this subsequence.
The idea is that it suffices to prove the converse result 
only for values of $n$ on an arbitrary subsequence of $(1,2,3,\dotsc)$, since 
we used the $\liminf$ in \eqref{eqn:second_def0} and the 
$\limsup$ in \eqref{eqn:second_def}.  

\subsubsection{A Verd\'u-Han-Type Converse Bound}

We make use of the following non-asymptotic converse bound from \cite{ScaTan13}:
\begin{align}
    \eps_n \ge 1 - \Pr\bigg(\frac{1}{n}\sum_{i=1}^{n}\bj(X_{1,i},X_{2,i},Y_i) \ge \bR_n - \gamma\bone \bigg) - 2e^{-n\gamma}, \label{eqn:vh}
\end{align}
where $\gamma$ is an arbitrary constant, $(\bX_1,\bX_2)$ is the random pair induced by
the codebook, and $\bY$ is the resulting output.  The output distributions defining $\bj$ 
are those induced by the fixed input joint type $P_{X_1 X_2, n}$.  By the above 
constant-composition reduction and a simple symmetry argument, we may replace 
$(\bX_1,\bX_2)$ by a fixed pair $(\bx_1,\bx_2) \in T^n(P_{X_1X_2,n})$.

\subsubsection{Handling Singular Dispersion Matrices}

Directly applying the multivariate Berry-Esseen theorem (e.g. see \cite[Sec. VI]{TK12}) 
to \eqref{eqn:vh} is problematic, since the dispersion matrix $\bV(P_{X_1X_2,n})$ may
be singular or asymptotically singular.  We therefore proceed by handling such matrices,
and reducing the problem to a lower dimension as necessary.

We henceforth use the shorthands $\bI_n := \bI(P_{X_1X_2,n})$ and $\bV_n := \bV(P_{X_1X_2,n})$.
An eigenvalue decomposition yields
\begin{align}
    \bV_n = \bU_n\bD_n\bU_n^T,
\end{align} 
where $\bU_n$ is unitary  (i.e. $\bU_n\bU_n^T$ is the identity matrix) and $\bD_n$ is diagonal.  
Since we passed to a convergent subsequence
in Step 3 and the eigenvalue decomposition map $\bV_n \to (\bU_n,\bD_n)$ is continuous, 
we conclude that both $\bU_n$ and $\bD_n$ converge, say to $\bU_{\infty}$ and $\bD_{\infty}$.
When $\rank(\bD_{\infty})=2$ (i.e. $\bD_{\infty}$ has full rank), we directly use the multivariate
Berry-Esseen theorem.  We proceed by discussing lower rank matrices.

Since $\bV_n$ is the covariance matrix of $\bA_n := \frac{1}{\sqrt n}\big( \sum_{i=1}^{n}\bj(x_{1,i},x_{2,i},Y_i) - n\bI_n \big)$
(with $Y_i\sim W(\cdot|x_{1,i},x_{2,i})$),
we see that $\bD_n$ is the covariance matrix of $\tilde{\bA}_n := \bU_n^T\bA_n$.  In the case that
$\rank(\bD_{\infty})=1$, we may write
\begin{align}
    \tilde{\bA}_n := \big[ \tilde{A}_{n,1} ~ \tilde{A}_{n,2}]^T,
\end{align}
where $\var[\tilde{A}_{n,1}]$ is bounded away from zero, and  $\var[\tilde{A}_{n,2}] \to 0$. 
Since $\bU_n$ is unitary, we have
\begin{align}
    \bA_n = \bU_n\tilde{\bA}_n  
  = \bU_{n,1}\tilde{A}_{n,1} + \bDelta_n, \label{eqn:An2}
\end{align}
where $\bU_{n,i}$ denotes the $i$-th column of $\bU_n$, and $\bDelta_n := \bU_{n,2}\tilde{A}_{n,2}$.
Since $\bA_n$ has mean zero by construction, the same is true of $\tilde{\bA}_n$ and hence $\bDelta_n$.
Moreover, since $\tilde{A}_{n,1}$ has vanishing variance, the same is true of each entry of $\bDelta_n$.
Thus, Chebyshev's inequality implies that, for any $\delta_n>0$,
\begin{align}
    \Pr\big(\|\bDelta_n\|_{\infty} \ge \delta_n \big) \le \frac{\psi_n}{\delta_n^2}, \label{eqn:Cheby}
\end{align} 
where $\psi_n := \max_{i=1,2}\var[\Delta_{n,i}] \to 0$.

We can now bound the probability in \eqref{eqn:vh} as follows:
\begin{align}
    & \Pr\bigg(\frac{1}{n}\sum_{i=1}^{n}\bj(X_{1,i},X_{2,i},Y_i) \ge \bR_n - \gamma\bone \bigg) \nn \\
        & \,\, = \Pr\Big( \bA_n \ge \sqrt{n}\big(\bR_n - \bI_n - \gamma\bone \big) \Big) \\
        & \,\, = \Pr\Big( \bU_{n,1}\tilde{A}_{n,1} + \bDelta_n \ge \sqrt{n}\big(\bR_n - \bI_n - \gamma\bone \big) \Big) \\
        & \,\, \le \Pr\Big( \bU_{n,1}\tilde{A}_{n,1} \ge \sqrt{n}\big(\bR_n - \bI_n - \gamma\bone\big) - \delta_n\bone \Big) \nn \\
        & \,\,\qquad\qquad\qquad\qquad\qquad   + \Pr\big(\|\bDelta_n\|_{\infty} \ge \delta_n \big) \\
        & \,\, \le \Pr\Big( \bU_{n,1}\tilde{A}_{n,1} \!\ge\! \sqrt{n}\big(\bR_n\! -\! \bI_n\! -\! \gamma\bone\big) - \delta_n\bone \Big) \!+\! \frac{\psi_n}{\delta_n^2}, \label{eqn:rank_one}
\end{align}
where the last three steps respectively follow from \eqref{eqn:An2}, \cite[Lemma 9]{TK12}, 
and \eqref{eqn:Cheby}.  We now choose $\delta_n=\psi_n^{1/3}$, so that both $\delta_n$ and 
$\frac{\psi_n}{\delta_n^2}$ are vanishing.  Equation \eqref{eqn:rank_one} permits an application
of the {\em univariate} Berry-Esseen theorem, since the variance of $\tilde{A}_{n,1}$ is bounded 
away from zero.

The case $\rank(\bD_{\infty})=0$ is handled similarly using Chebyshev's inequality, and we thus 
omit the details and merely state that \eqref{eqn:rank_one} is replaced by
\begin{align}
    \openone\Big( \sqrt{n}\big(\bR_n - \bI_n - \gamma\bone\big) \le \delta_n\bone \Big) + \delta_n' \label{eqn:rank_zero}
\end{align}  
where $\delta_n \to 0$ and $\delta_n' \to 0$.

\subsubsection{Application of the Berry-Esseen Theorem}

Let $\bI_{\infty}$ and $\bV_{\infty}$ denote the limiting values (on the convergent
subsequence of block lengths) of $\bI_n$ and $\bV_n$. In this step, we will use the fact 
that $\Psi^{-1}(\cdot,\eps)$ is continuous in the following sense:
\begin{align}
    \Psi^{-1}(\bV_n,\eps) - \delta\bone \subset \Psi^{-1}(\bV_{\infty},\eps) \subset \Psi^{-1}(\bV_n,\eps) + \delta\bone \label{eqn:continuity} 
\end{align}
for any $\delta>0$ and sufficiently large $n$.  This is proved using a 
Taylor expansion when $\bV_{\infty}$ has full rank, and is proved similarly to 
\cite[Lemma 6]{ScaTan13} when $\bV_{\infty}$ is singular.

We claim that the preceding two steps, along with the choice $\gamma := \frac{\log n}{n}$,
imply that the rate pair $(R_{1,n},R_{2,n})$ satisfies
\begin{equation}
    \bR_n \in \bI_n + \frac{1}{\sqrt n}\Psi^{-1}(\bV_{\infty},\eps) + g(n)\bone \label{eq:global_conv}
\end{equation}
for some $g(n)=o\big(\frac{1}{\sqrt n}\big)$ depending on $P_{X_1,X_2,n}$ and $\eps$.  
In the case $\rank(\bD_{\infty})=2$ (see the preceding step), this follows by applying the multivariate
Berry-Esseen theorem with a positive definite covariance matrix, re-arranging to obtain
\eqref{eq:global_conv} with $\bV_n$ in place of $\bV_{\infty}$, and then using \eqref{eqn:continuity}.

In the case $\rank(\bV_{\infty})=1$, we obtain \eqref{eq:global_conv} by applying the univariate 
Berry-Esseen theorem to \eqref{eqn:rank_one} and similarly applying rearrangements and 
\eqref{eqn:continuity}.  The resulting expression can be written in the multivariate form
in \eqref{eq:global_conv} by a similar argument to \cite[p. 894]{TK12}.

When $\rank(\bV_{\infty})=0$, we have $\bV_{\infty}=\bzero$, and $\Psi^{-1}(\bV_{\infty},\eps)$ is simply the 
quadrant $\{\bz \,:\, \bz \le \bzero \}$.  We thus obtain \eqref{eq:global_conv} by
noting that the indicator function in \eqref{eqn:rank_zero} is zero for sufficiently
large $n$ whenever either entry of $\bR_n$ exceeds the corresponding entry of $\bI_n$
by $\Theta\big(\frac{1}{\sqrt n}\big)$. 

\subsubsection{Establishing the Convergence to $\Pi(R_1^*,R_2^*)$}

We use a proof by contradiction to show that the limiting value $P_{X_1X_2,\infty}$
of $P_{X_1X_2,n}$ (on the convergent subsequence of block lengths) must lie
within $\Pi(R_1^*,R_2^*)$.  Assuming the contrary, we observe from
\eqref{eqn:def_Pi} that at least one of the strict inequalities
$I_1(P_{X_1X_2,\infty}) < R_1^*$ and $I_{12}(P_{X_1X_2,\infty}) < R_1^* + R_2^*$ must hold.
It thus follows from \eqref{eq:global_conv} and the continuity of $\bI(P_{X_1X_2})$ that there exists $\delta>0$
such that either $R_{1,n} \le R_1^* - \delta$ or $R_{1,n} + R_{2,n} \le R_1^* + R_2^* - \delta$
for sufficiently large $n$, in contradiction with the convergence of $(R_{1,n},R_{2,n})$
to $(R_1^*,R_2^*)$ implied by \eqref{eqn:second_def0}.

\subsubsection{Completion of the Proof for Cases (i) and (ii)}

Here we handle distributions $P_{X_1X_2,\infty}$ corresponding to the
cases in \eqref{eqn:L0_1}--\eqref{eqn:L0_12}.  We focus on case (ii), 
since case (i) is handled similarly.  

It is easily verified from \eqref{eqn:psiinv} that each point $\bz$ in $\Psi^{-1}(\bV,\eps)$ 
satisfies $z_1 + z_2 \le \sqrt{V_{12}}\Phi^{-1}(\eps)$.  We can thus weaken \eqref{eq:global_conv} to
\begin{align}
    R_{1,n}+R_{2,n} \le I_{12}(P_{X_1X_2,n}) + \sqrt{\frac{V_{\infty,12}}{n}}\Phi^{-1}(\eps) + g(n).
\end{align}
We will complete the proof by showing that 
$I_{12}(P_{X_1X_2,n}) \le R_1^* + R_2^*$ for all $n$. 
Since $\bigcup_{P_{X_1X_2}}\big\{I_{12}(P_{X_1X_2})\big\}$ is the set of all achievable (first-order) sum rates,
it suffices to show that any boundary point corresponding to \eqref{eqn:L0_12} is
one maximizing the sum rate.  We proceed by establishing that this is true.

The conditions stated before \eqref{eqn:L0_12} state that $(R_1^*,R_2^*)$ lies on
the diagonal part of the achievable trapezium corresponding to $P_{X_1X_2}$, and away
from the corner point.  It follows that $\bp_1 := (R_1^* - \delta,R_2^* + \delta)$ and
$\bp_2 := (R_1^* + \delta,R_2^* - \delta)$ are achievable for sufficiently small $\delta$.
If another point $\bp_0$ with a strictly higher sum rate were achievable, then all points
within the triangle with corners defined by $\bp_0$, $\bp_1$ and $\bp_2$ would also 
be achievable.  This would imply the achievability of $(R_1^* + \delta',R_2^* + \delta')$
for sufficiently small $\delta'>0$, which contradicts the assumption that $(R_1^*,R_2^*)$ is a boundary point.


\subsubsection{Completion of the Proof for Case (iii))}

We now turn to the remaining case in \eqref{eqn:L0_both}, corresponding to 
$\bI_{\infty} = \bR^*$.  Again using the fact that $\bI(P_{X_1X_2})$ 
is continuous in $P_{X_1X_2}$, we have 
\begin{align}
    \bI_n = \bR^* + \bDelta(P_{X_1X_2,n}), \label{eqn:Idiff}
\end{align}
where $\|\bDelta(P_{X_1X_2,n})\|_{\infty} \to 0$.  We claim 
that $\bDelta(P_{X_1X_2,n}) \in \calV(R_1^*,R_2^*)$ 
(see \eqref{eqn:setV}). Indeed, if this were not the case, then 
\eqref{eqn:Idiff} would imply that the pair $(I_{n,1},I_{n,12}-I_{n,1})$ 
lies outside the capacity region, in contradiction with the fact
that rates satisfying \eqref{eqn:cap_region1}--\eqref{eqn:cap_region2}
are (first-order) achievable for {\em all} $P_{X_1X_2}$.

Assuming for the time being that $\|\bDelta(P_{X_1X_2,n})\|_{\infty} = O\big(\frac{1}{\sqrt n}\big)$,
we immediately obtain the outer bound
\begin{align}
    & \calL(\eps,R_{1}^{*},R_{2}^{*})\subseteq\bigg\{(L_{1},L_{2})\,:\, \nn \\
    & \bL\in\bigcup_{P_{X_{1}X_{2}}\in\Pi,\bT\in\calV}\Big\{\Psi^{-1}(\bV(P_{X_{1}X_{2}}),\eps)+\bT\Big\}\bigg\}. \label{eqn:alt_form}
\end{align}
The set in \eqref{eqn:alt_form} clearly includes $\calL_0$
in \eqref{eqn:L0_both}.  We proceed by showing that the reverse inclusion
holds, and hence the two sets are identical.  Since $\hat{\bT}_{-}$
and $\hat{\bT}_{+}$ are tangent vectors, any vector $\bT\in\calV$
can have one or more of its components increased to yield a vector $\bT'$
whose direction coincides with either $\bT_{-}$ or $\bT_{+}$.  The fact
that the magnitude of $\bT'$ may be arbitrary is captured by the unions
over $\beta\ge0$ in \eqref{eqn:L0_both}.

It remains to handle the case that 
$\|\bDelta(P_{X_1X_2,n})\|_{\infty}$ is not $O\big(\frac{1}{\sqrt n}\big)$. 
By performing another pass to a subsequence of block lengths if necessary, we can assume that
$\|\bDelta(P_{X_1X_2,n})\|_{\infty} = \omega\big(\frac{1}{\sqrt n}\big)$.
Such scalings can be shown to play no role in characterizing $\calL$, similarly to
\cite{ScaTan13}; we provide only an outline here.  
Let $\Delta_{n,1}$ and $\Delta_{n,12}$ denote the entries of $\bDelta(P_{X_1X_2,n})$,
and let $\Delta_{n,2} := \Delta_{n,12} - \Delta_{n,1}$.  If either $\Delta_{n,1}$
or $\Delta_{n,2}$ is negative and decays with a rate $\omega\big(\frac{1}{\sqrt n}\big)$, then
no value of the corresponding $L_j$ ($j\in\{1,2\}$) can satisfy 
the condition in \eqref{eqn:second_def0}, so the converse is trivial.  
On the other hand, if either $\Delta_{n,1}$ or $\Delta_{n,2}$ is 
positive and $\omega\big(\frac{1}{\sqrt n}\big)$, we simply recover the right-hand side
of \eqref{eqn:alt_form} in the limiting case that either $T_1$ or
$T_{12} - T_1$ (where $\bT := [T_1~~T_{12}]^T$) grows unbounded.  Thus,
the required converse statement for this case is already captured by \eqref{eqn:alt_form}. 

 
\subsubsection*{Acknowledgments}
The second author  is supported by NUS startup grants WBS R-263-000-A98-750/133. 
\bibliographystyle{unsrt}
\bibliography{isitbib}

\begin{thebibliography}{10}

\bibitem{Strassen}
V.~Strassen.
\newblock {Asymptotische Absch\"{a}tzungen in Shannons Informationstheorie}.
\newblock In {\em Trans. Third Prague Conf. Inf. Theory}, pages 689--723,
  Prague, 1962.

\bibitem{Hayashi09}
M.~Hayashi.
\newblock Information spectrum approach to second-order coding rate in channel
  coding.
\newblock {\em IEEE Trans. Inf. Th.}, 55(10):4947--66, 2009.

\bibitem{PPV10}
Y.~Polyanskiy, H.~V. Poor, and S.~Verd\'{u}.
\newblock Channel coding rate in the finite blocklength regime.
\newblock {\em IEEE Trans. Inf. Th.}, 56(5):2307--59, 2010.

\bibitem{TK12}
V.~Y.~F. Tan and O.~Kosut.
\newblock On the dispersions of three network information theory problems.
\newblock {\em IEEE Trans. Inf. Th.}, 60(2):881--903, 2014.

\bibitem{Nom14}
R.~Nomura and T.~S. Han.
\newblock Second-order {Slepian-Wolf} coding theorems for non-mixed and mixed
  sources.
\newblock {\em IEEE Trans. Inf. Th.}, 60(9):5553--5572, 2014.

\bibitem{LeTan14}
S.-Q. Le, V.~Y.~F. Tan, and M.~Motani.
\newblock Second-order asymptotics for the {Gaussian} interference channel with
  strictly very strong interference.
\newblock In {\em Int. Symp. Inf. Th.}, Honolulu, HI, July 2014.

\bibitem{ScaTan13}
J.~Scarlett and V.~Y.~F. Tan.
\newblock Second-order asymptotics for the {G}aussian {MAC} with degraded
  message sets.
\newblock {\tt arXiv:1310.1197 [cs.IT]}.

\bibitem{elgamal}
A.~{El~Gamal} and Y.-H. Kim.
\newblock {\em Network Information Theory}.
\newblock Cambridge University Press, Cambridge, U.K., 2012.

\bibitem{Csi97}
I.~Csisz\'{a}r and J.~{K\"{o}rner}.
\newblock {\em Information Theory: Coding Theorems for Discrete Memoryless
  Systems}.
\newblock Cambridge University Press, 2011.

\bibitem{Scarlett13b_arxiv}
J.~Scarlett, A.~Martinez, and A.~{Guill\'en i F\`abregas}.
\newblock Second-order rate region of constant-composition codes for the
  multiple-access channel.
\newblock {\em IEEE Trans. Inf. Th.}, 61(1):157--172, Jan. 2015.

\end{thebibliography}
 
\end{document}